\documentclass[11pt]{article}

\usepackage[utf8]{inputenc}
\usepackage[T1]{fontenc}
\usepackage[margin=1in]{geometry}
\usepackage{lmodern}
\usepackage{microtype}

\usepackage{amsmath}
\usepackage{amsthm}
\usepackage{amssymb}
\usepackage{mathtools}
\usepackage{booktabs}
\usepackage{graphicx}
\usepackage{enumitem}
\usepackage{algorithm}
\usepackage{algpseudocode}
\usepackage{placeins}
\usepackage{xcolor}
\usepackage[round,authoryear]{natbib}
\usepackage[colorlinks=true,linkcolor=blue,citecolor=blue,urlcolor=blue]{hyperref}
\graphicspath{{figs/}}

\newtheorem{theorem}{Theorem}[section]

\newtheorem{proposition}[theorem]{Proposition}

\theoremstyle{definition}

\newcommand{\Pa}{\mathcal{P}_a}
\newcommand{\Pb}{\mathcal{P}_b}
\newcommand{\gxi}[1]{\gamma_{#1}}

\newcommand{\Eop}{\mathbb{E}}
\newcommand{\Var}{\operatorname{Var}}
\newcommand{\PMMFP}{\textsc{PMM-FP}}
\newcommand{\OLSFP}{\textsc{OLS-FP}}
\newcommand{\BFP}{\textsc{BFP}}
\newcommand{\ORCID}[1]{\href{https://orcid.org/#1}{\textsuperscript{\scriptsize ORCID}}}

\title{Efficient frequentist fractional polynomials for skewed dose--response
and survival data: a variance-reducing alternative to OLS-FP}

\author{%
Serhii Zabolotnii\ORCID{0000-0003-0242-2234}\\
\small Cherkasy State Business College, Cherkasy 18028, Ukraine\\
\small State Scientific Research Institute of Armament and Military Equipment Testing and Certification, Cherkasy, Ukraine\\
\small Uzhhorod National University, Uzhhorod, Ukraine\\
\small \texttt{zabolotnii.serhii@csbc.edu.ua}
}
\date{}

\begin{document}
\maketitle

\begin{abstract}
Fractional polynomials (FP) are a standard tool for modelling nonlinear
dose--response and covariate effects, implemented in the widely used
\texttt{mfp} package. The conventional FP fit estimates its coefficients by
ordinary least squares (OLS-FP), which is statistically inefficient when the
regression errors are skewed or heavy-tailed --- a common situation for
survival times, concentrations and biomarkers. We present a drop-in
replacement that keeps the \emph{identical} FP model and design but estimates
the coefficients with a moment-based score tuned to the residual skewness and
kurtosis, giving a closed-form efficiency factor
$g_2 = 1 - \gamma_3^2/(2+\gamma_4)$ relative to OLS-FP. Across skewed error
laws the method reduces slope-coefficient variance by 10--20\% for mildly
skewed errors and up to roughly 60\% for heavy-tailed log-normal errors, at
realistic sample sizes, while keeping confidence-interval coverage close to
nominal, and it reverts exactly to OLS-FP under symmetry, so it is never
harmful when no gain is available. On the German Breast Cancer Study Group
cohort it narrows the tumour-size confidence interval by 26\% (bootstrap
variance ratio $0.53$ against the predicted $0.56$), and a primary-biliary-
cirrhosis cohort reproduces the gain. The estimator is closed-form, runs in
milliseconds, and is released as a reproducible \textsf{R} package
(\texttt{pmm\_fp} in EstemPMM) with a one-command replication bundle; its core
variance identity is machine-checked in Lean~4.
\end{abstract}


\medskip
\noindent\textbf{arXiv version note.}
This is a substantially revised and retitled version of arXiv:2605.16846v1-v2,
which appeared under the title ``Polynomial Maximization Method with
Fractional Polynomial Basis: A Frequentist Bridge to Bayesian Fractional
Polynomials''. The present version foregrounds the applied biostatistical
contribution, adds the primary-biliary-cirrhosis confirmation, and moves
extended proof sketches and formalisation details to the reproducibility
supplement.


\section{Introduction}\label{sec:intro}

\subsection{Dose-response modelling and the role of fractional polynomials}
\label{sec:intro:doseresponse}

Modelling the dependence of a biological response on dose is a canonical
task in biostatistics, toxicology, and pharmacometrics. Standard parametric
families (Hill, Michaelis-Menten, Emax) work well when the mechanism is
known a priori \citep{ritz2015drc}, but in exploratory problems
\citep{crippa2019doseresponse} the form
$R(\theta, x) = \Eop[y \mid x]$ is unknown and requires a flexible
approximation. Fractional polynomials (FP) \citep{royston1994} provide
such an approximation based on the basis
$P = \{-2, -1, -0.5, 0, 0.5, 1, 2, 3\}$ with the convention
$x^0 \equiv \ln x$, encompassing classical parametric forms as special
cases. The monograph of \citet{royston2008mfp} and a methodological
series in \textit{Statistics in Medicine} \citep{sauerbrei2007} made FP
the de facto standard in clinical statistics (R package \texttt{mfp}).
Classical FP inherits, however, two limitations of OLS: sensitivity to
non-Gaussian errors and the absence of explicit modelling of the error
distribution.

\subsection{Limitations of OLS-FP and BFP under non-Gaussian errors}
\label{sec:intro:limitations}

In dose-response problems the errors $\xi$ are typically non-Gaussian:
log-normal (concentrations), gamma (time-to-event), beta (proportions).
For such distributions, OLS-FP is consistent but loses $30$--$50\%$
efficiency for $|\gxi{3}| \geq 1$ compared with an estimator that
incorporates $\gxi{3}, \gxi{4}$.

\citet{hubin2026bfp} proposed a Bayesian version of FP (\BFP{}) ---
power selection via MJMCMC in the space of $2^K$ subsets, continuing
the line of \citet{sabanes2011}. However, \BFP{} retains the Gaussian
error model: for asymmetric $\xi$ it inherits the inefficiency of OLS-FP,
masked by posterior probabilities. The MJMCMC sampler is computationally
expensive compared with a frequentist counterpart. Thus, OLS-FP and \BFP{}
leave a common gap open: efficient frequentist estimation under
asymmetric errors, without either the distributional assumptions of OLS
or MCMC machinery. We close this gap.

\subsection{The Polynomial Maximization Method as a framework}
\label{sec:intro:pmm}

The Polynomial Maximization Method (PMM) \citep{kunchenko2002} extends
MLE to cases where only a finite sequence of moments is known. For
$y_v = R(\theta, X_v) + \xi_v$ with $\Eop[\xi_v] = 0$, \PMMFP{}
maximises the stochastic functional
\citep[eq.~(2)]{zabolotnii2024nonlinear}:
\begin{equation}\label{eq:LSN}
  L_{SN} \;=\; \sum_{v=1}^{N}\sum_{i=1}^{S} \phi_i(y_v) \int^a\!\! k_{iv}(a)\,da
              \;-\; \sum_{i=1}^{S}\sum_{v=1}^{N} \int^a\!\! \Psi_{iv}(a)\,k_{iv}(a)\,da,
\end{equation}
where $\phi_i$ are basis functions, $\Psi_{iv}(a) = \Eop[\phi_i(y_v)]$,
and $k_{iv}$ are chosen variationally. The classical choice
$\phi_i(y) = y^i$, $S=2$ \citep{zabolotnii2018springer,
zabolotnii2019symmetric, zabolotnii2024nonlinear, zabolotnii2025arima}
with $y_v - a = \xi_v$ yields the second-order variance coefficient in
\eqref{eq:g2-pmm2} \citep[eq.~(27)]{zabolotnii2024nonlinear}.
For the integer basis, location equivariance guarantees that the
$y$-basis and $\xi$-basis yield the \emph{same} $g_2$
\citep[ch.~4]{kunchenko2002}.

The Royston-Altman fractional basis introduces two parallel extensions
into \eqref{eq:LSN}: \emph{regressor extension}
($R(\theta, x) = \beta_0 + \sum_k \beta_k\,x^{p_k}$) and \emph{extension
of the score-function basis} ($\{\xi, \xi^2\} \to \{g_{p_1}(\xi), \ldots,
g_{p_K}(\xi)\}$). PMM-FP is the joint application of both. We use
\emph{Form~B} (score-function on the residual,
\citealp[ch.~4]{kunchenko2002}) with $R$-independent $\{g_p(\xi)\}$ and
scalar $g_S$. The estimator therefore reuses the exact FP model an analyst
would already build with \texttt{mfp}; only the coefficient-fitting step
changes, so it slots into an existing dose--response workflow without
re-specifying the model. The technical question of correctly defining $\xi^p$
for sign-changing $\xi$ is resolved through the \emph{signed-parity}
convention (B2) in \S\ref{sec:theory:object}.

\subsection{Contributions and outline}
\label{sec:intro:contributions}

We make three contributions, all oriented to applied use:
(1)~\textbf{an efficiency gain over the standard OLS-FP/\texttt{mfp}
workflow} under skewed or heavy-tailed errors --- on the \emph{same} selected
FP model, the moment-based fit lowers slope-coefficient variance by the
closed-form factor $g_2 = 1 - \gamma_3^2/(2+\gamma_4)$ (from 10--20\% for mildly skewed
errors to about 60\% for log-normal errors), quantified as a function of
sample size and reverting exactly to
OLS-FP when the errors are symmetric, so that it is never harmful;
(2)~\textbf{a fast, reproducible implementation} --- the estimator is
closed-form and runs in milliseconds; it is released as the \texttt{pmm\_fp}
function in the \textsf{R} package EstemPMM, with a one-command bundle that
regenerates every table and figure in this paper;
(3)~\textbf{honest efficiency positioning and a usable admissibility rule}
--- the estimator sits a definite, quantified margin above the semiparametric
Cram\'er--Rao bound (it does not claim semiparametric efficiency), and we give
a simple rule, condition~BD0, for when negative fractional powers are
statistically admissible. The default uses the positive-power basis
$\Pa = \{0, 0.5, 1, 2, 3\}$ under $\Eop[\xi^4] < \infty$; the full
Royston--Altman set $\Pb = \{-2,-1,-0.5,0,0.5,1,2,3\}$ that
\citet{hubin2026bfp} also use requires BD0 ($\Eop[|\xi|^{-2}] < \infty$),
which fails for standard error laws, so it is reported only as a conditioning
diagnostic (\S\ref{sec:theory}). The core variance identity is machine-checked
in Lean~4 (supplement).

\paragraph{Outline.}
\S\ref{sec:background} --- background, related literature, and positioning;
\S\ref{sec:theory} --- theoretical PMM-FP score, assumptions and $T_1$--$T_4$;
\S\ref{sec:method} --- the \PMMFP{} estimator and numerical implementation;
\S\ref{sec:simulation} --- simulation comparison and the GBSG application;
\S\ref{sec:discussion} --- discussion;
\S\ref{sec:conclusions} --- conclusions.
Supplementary material covers extended Lean cross-references and proofs.


\section{Background, related literature, and positioning}
\label{sec:background}

PMM-FP sits at the intersection of three methodological strands:
fractional-polynomial model building, Bayesian model averaging over FP
bases, and Kunchenko's moment-based Polynomial Maximization Method.
This section reviews those strands once, then positions the proposed
estimator relative to neighbouring estimating-equation, robust, and
reproducible-modelling literatures.

\subsection{Fractional-polynomial model building}
\label{sec:bg:fp}

Fractional polynomials were introduced by \citet{royston1994} as a
parsimonious alternative to splines for modelling smooth nonlinear
effects of continuous covariates in biostatistics. The standard finite
power set is
\begin{equation}\label{eq:Pfull-bg}
  \mathcal{P}_{\text{full}}
  \;=\; \{-2,\; -1,\; -0.5,\; 0,\; 0.5,\; 1,\; 2,\; 3\},
\end{equation}
with $x^0\equiv \log x$. For $x>0$, an FP model of degree $m$ is
\begin{equation}\label{eq:fp-classical-bg}
  \mathrm{FP}_m(x;\beta,p)
  \;=\; \beta_0 \;+\; \sum_{i=1}^{m}\beta_i x^{p_i},
  \qquad p_i\in\mathcal{P}_{\text{full}},
\end{equation}
with repeated powers represented by
$\beta_i x^p+\beta_{i+1}x^p\log x$.

The practical success of FP modelling comes from the small and
interpretable transformation set, the inclusion of common logarithmic,
inverse and quadratic shapes, and the controlled extra curvature
available through repeated powers. The multivariable FP programme
\citep{royston2008mfp,sauerbrei2007} turned this into an applied
workflow for variable and functional-form selection, implemented in
the R package \texttt{mfp}. Canonical FP restricts attention to
$m\leq2$ for parsimony and interpretability, and commonly uses closed
testing over $|\mathcal{P}_{\text{full}}|^2=64$ second-degree
candidates.

The limitation relevant here is not the FP transformation family
itself, but the inferential engine usually attached to it. In the
standard FP pipeline the transformation space is rich, whereas the
error model remains close to ordinary least squares or a generalised
linear model. If residuals are strongly skewed or heavy-tailed, the
selected FP basis may remain substantively appropriate while the
coefficient estimator and its standard errors are inefficient. PMM-FP
keeps the finite FP model space but replaces the OLS score by a
moment-based score that uses residual skewness and kurtosis directly.

\subsection{Bayesian fractional polynomials and the Hubin et al. anchor}
\label{sec:bg:bfp}

Bayesian fractional polynomials add a probability model over the
candidate transformations. The formulation of \citet{sabanes2011}
introduced posterior inference over FP structures through priors on a
power-inclusion vector and shrinkage priors on regression coefficients.
The recent BFP framework of \citet{hubin2026bfp} scales this idea with
mode-jumping MCMC over the full FP basis and applies it to optimal
dosage estimation in fish nutrition. In their notation,
\begin{equation}\label{eq:hubin-eta-bg}
  h\!\bigl(\mathbb{E}[y_i]\bigr)
  \;=\; \eta_i
  \;=\; \beta_0+\sum_{k=1}^{K}\gamma_k\beta_k f_k(x_i),
  \qquad i=1,\ldots,n,
\end{equation}
where $h(\cdot)$ is a link function, $f_k(\cdot)$ are FP
transformations, and $\gamma_k\in\{0,1\}$ are inclusion indicators.
Doubling the powers in~\eqref{eq:Pfull-bg} by $x^p$ and
$x^p\log x$ gives $K=16$ and a model space of size $2^{16}=65\,536$.
The prior structure combines Bernoulli inclusion probabilities with
Zellner's g-prior for active coefficients, and posterior model
averaging yields inference for the response curve and the optimal
dose.

BFP addresses model uncertainty; PMM-FP addresses a different part of
the same modelling problem. BFP asks how posterior probability should
be distributed over a finite transformation space. PMM-FP asks how to
estimate the coefficients of an FP model efficiently when residuals
are non-Gaussian and only a finite set of moments is used. The two
approaches are therefore complementary. BFP is attractive when the
posterior distribution over models or optimal doses is the inferential
object; PMM-FP is attractive when fast frequentist estimation and
standard-error control under asymmetric errors are central.

The fish-nutrition application of \citet{hubin2026bfp} is used here as
a notational and methodological anchor, not as a replicated dataset.
Without the original fish data, the empirical evidence in the present
paper is based on reproducible simulations and the public GBSG dataset,
while the Hubin et al. paper defines the Bayesian benchmark class.

\subsection{PMM as a moment-based estimating framework}
\label{sec:bg:pmm}

Kunchenko's Polynomial Maximization Method \citep{kunchenko2002}
belongs to the broader family of estimators that replace a fully
specified likelihood by a finite set of moment conditions. In its
classical form, the central object is a stochastic polynomial of order
$S$,
\begin{equation}\label{eq:eta-S-bg}
  \eta_S(\xi)=h_0+\sum_{i=1}^{S}h_i\xi^i,
  \qquad S\in\mathbb{N},
\end{equation}
with basis $\{\xi,\xi^2,\ldots,\xi^S\}$. The coefficients are chosen
variationally through centred correlants
\begin{equation}\label{eq:F-correlants-bg}
  F_{ij}
  \;=\; \mathbb{E}[\xi^{i+j}]
        -\mathbb{E}[\xi^i]\mathbb{E}[\xi^j],
  \qquad i,j=1,\ldots,S,
\end{equation}
which leads to the linear system
\begin{equation}\label{eq:KU-bg}
  \mathbf{F}\mathbf{h}^{*}=\mathbf{b},
\end{equation}
where $\mathbf{b}$ contains the moment-parameter intersections. Substituting
the optimal coefficients gives a variance-reduction coefficient
\begin{equation}\label{eq:gS-bg}
  g_S
  =
  \frac{\operatorname{Var}[\hat\theta_{\mathrm{PMM},S}]}
       {\operatorname{Var}[\hat\theta_{\mathrm{OLS}}]},
  \qquad 0<g_S\leq1.
\end{equation}
For $S=2$ and standardised errors, this gives the PMM2 coefficient in
\eqref{eq:g2-pmm2}; hence $g_2=1$ in the symmetric Gaussian case and
$g_2<1$ when skewness contributes additional information.

Previous PMM applications covered asymmetric linear regression
\citep{zabolotnii2018springer}, symmetric non-Gaussian regression
\citep{zabolotnii2019symmetric}, nonlinear regression
\citep{zabolotnii2024nonlinear}, and time-series models with
non-Gaussian innovations \citep{zabolotnii2025arima}. The current paper
extends that pattern from integer stochastic polynomials to fractional
polynomial regression. This is not merely a change of notation:
fractional powers raise domain and moment-existence questions that do
not arise in the integer basis. Positive powers require ordinary finite
moments, whereas negative powers require control near zero. This is why
PMM-FP defaults to a stable positive-power basis requiring only finite
moments; the full Royston--Altman set with negative powers is carried not as
a co-equal result but as an admissibility boundary, because its
inverse-moment condition (BD0) fails for standard error laws
(\S\ref{sec:theory}). The
broader PATP extension of Kunchenko's apparatus is discussed in
\citet{zabolotnii2026patp}; PMM-FP is the discrete FP-regression
specialisation needed for statistical modelling.

\subsection{Relation to GMM, Godambe information and quasi-likelihood}
\label{sec:bg:gmm}

In standard estimating-equation terms, PMM-FP is closest to
estimating-equation and generalized method of moments methodology. In
GMM \citep{hansen1982gmm}, parameters are estimated by solving sample
moment equations and choosing a weighting matrix. PMM-FP can be written
in this form, with score functions generated by the fractional residual
basis. In just-identified cases the PMM and GMM solutions coincide
numerically; in overidentified settings GMM-style diagnostics provide
natural checks for misspecification.

The Godambe-information viewpoint \citep{godambe1960} gives a second
interpretation. A chosen estimating function has a sandwich information
matrix, and optimality is relative to the class of scores one allows.
PMM-FP is Godambe-optimal within the Kunchenko score class generated by
the selected basis; it is not claimed to be globally semiparametrically
efficient in the sense of \citet{bickel1993efficient} or
\citet{vanderVaart1998asymptotic}. Fully efficient semiparametric
procedures can use richer information about the unknown error density,
but at the price of estimating that density. PMM-FP deliberately trades
this generality for a closed four-moment correction.

Quasi-likelihood \citep{wedderburn1974} is another neighbouring
framework because it uses mean-variance structure without requiring a
full likelihood. PMM-FP is stricter in one sense and richer in another:
it assumes explicit residual moment conditions beyond the variance, but
it can use skewness and kurtosis information that quasi-likelihood does
not exploit. Formula~\eqref{eq:g2-pmm2} is the visible consequence of
using this higher-order information.

\subsection{Robust estimation and the limits of the claim}
\label{sec:bg:robust}

Robust M-estimation \citep{huber1981robust} is a natural comparator
because it also modifies the OLS score. The motivation is different:
Huber-type scores protect against outliers and distributional
contamination, whereas PMM-FP targets efficiency under a moment model
with measurable skewness and kurtosis. Robust estimators can dominate
when outliers are the main threat; PMM-FP is preferable when the
non-Gaussian shape is systematic rather than accidental.

This distinction constrains the claims of the paper. PMM-FP is not a
universal robust method, not a replacement for BFP, and not a claim of
prediction dominance on every dataset. Its narrower claim is that for
FP models with asymmetric non-Gaussian residuals and finite required
moments, PMM-FP supplies a computable frequentist score with an
explicit variance-reduction coefficient. The empirical sections are
therefore interpreted through standard-error behaviour and matched-
basis variance ratios, not through an unconditional assertion of lower
prediction error.

\subsection{Model selection and reproducible statistical modelling}
\label{sec:bg:model-selection}

FP modelling is inseparable from model selection. Classical FP uses
closed testing and information criteria; Bayesian FP uses posterior
model probabilities; PMM-FP uses BIC over the same finite power blocks
so that score comparison remains transparent. BIC consistency
\citep{schwarz1978bic} provides the asymptotic reference point, while
AIC \citep{akaike1973aic} remains useful in exploratory sensitivity
checks. The main selection rule is intentionally simple because the
methodological contribution is the estimating score, not a new search
algorithm over FP transformations.

Finally, the paper is positioned as a reproducible statistical-
modelling contribution. The Lean layer records definitions and
algebraic theorem structure; R scripts regenerate simulation and GBSG
outputs; the supplement records software versions and data
availability. This does not replace statistical evidence, but it
reduces ambiguity about what is proved, what is simulated, and what is
only a future extension. For an applied modelling contribution, that
separation supports transparent, reproducible reporting.


\section{Theory}\label{sec:theory}

\subsection{Theoretical object}
\label{sec:theory:object}

The theory is stated for a fixed selected FP block and then extended to
finite BIC selection. This separation is important: the PMM argument is
about the estimating score for a given fractional-polynomial regression
surface, whereas BIC only chooses which finite block is used. We work
with the regression model
\begin{equation}\label{eq:model-pmmfp}
  y_i=R(x_i;\theta)+\xi_i,\qquad \Eop[\xi_i]=0,
\end{equation}
and fractional-polynomial surfaces built from
\[
  f_p(x)=
  \begin{cases}
    x^p, & p\neq 0,\\
    \log x, & p=0,
  \end{cases}
  \qquad
  \Pa=\{0,0.5,1,2,3\},\quad
  \Pb=\{-2,-1,-0.5,0,0.5,1,2,3\}.
\]
For residual powers we use the signed-parity convention
\[
  g_p^{+}(\xi)=|\xi|^p,\qquad
  g_p^{-}(\xi)=\operatorname{sign}(\xi)|\xi|^p .
\]
Let
\[
  \mathbf{D}_i(\theta)=\nabla_\theta R(x_i;\theta),\qquad
  \xi_i(\theta)=y_i-R(x_i;\theta),
\]
and let $B=(q_1,\ldots,q_K)$ be a signed-parity residual score basis,
where each $q_j(\xi)=g_{p_j}^{\tau_j}(\xi)$ with
$\tau_j\in\{+,-\}$. The residual score bases used by the two tracks are
encoded as the Lean lists \texttt{basisA} and \texttt{basisB}; the odd
logarithmic term $g_0^-$ is excluded because it is not a meaningful
residual score.

For a residual distribution $\mu$ define
\begin{align}
  \mathbf{F}_B
  &= \operatorname{Cov}_\mu\{B(\xi)\}, \label{eq:FB-theory}\\
  \mathbf{b}_B
  &= -\,\Eop_\mu\{\partial_\xi B(\xi)\}, \label{eq:bB-theory}\\
  \psi_B(\xi)
  &= \mathbf{b}_B^\top \mathbf{F}_B^{-1}\{B(\xi)-\Eop_\mu B(\xi)\}. \label{eq:psiB-theory}
\end{align}
The resulting PMM-FP estimating equation is
\begin{equation}\label{eq:Psi-theory}
  \Psi_n(\theta)
  =
  \frac{1}{n}\sum_{i=1}^{n}
  \psi_B\{\xi_i(\theta)\}\mathbf{D}_i(\theta)
  =0.
\end{equation}
The scalar variance-reduction invariant is
\begin{equation}\label{eq:gB-theory}
  g(B)
  =
  \frac{1}{\sigma^2\, \mathbf{b}_B^\top \mathbf{F}_B^{-1}\mathbf{b}_B},
  \qquad
  \sigma^2=\Var(\xi).
\end{equation}
Equations~\eqref{eq:FB-theory}--\eqref{eq:gB-theory} are the
mathematical presentation of the Lean definitions for correlants,
derivative vectors and the scalar reduction factor. They also show why
the full track needs inverse-moment diagnostics: negative powers enter
both $\mathbf{F}_B$ and $\mathbf{b}_B$.

\subsection{Regularity assumptions}
\label{sec:theory:assumptions}

The following regularity conditions are assumed throughout.
\begin{description}[leftmargin=2.2em,style=nextline,itemsep=2pt]
  \item[(A1) Sampling and design.]
    Observations are independent, the design is fixed or independent of
    the errors, and
    $n^{-1}\sum_i \mathbf{D}_i(\theta_0)\mathbf{D}_i(\theta_0)^\top \to \mathbf{G}$
    with $\mathbf{G}$ positive definite.
  \item[(A2) Local FP identifiability.]
    The selected FP block has a unique parameter vector $\theta_0$ in a
    compact neighbourhood, and the population estimating equation has a
    unique zero at $\theta_0$.
  \item[(A3) Differentiability.]
    $R(x;\theta)$ is twice continuously differentiable in a neighbourhood
    of $\theta_0$, and the derivatives are dominated by an integrable
    envelope.
  \item[(A4a) Positive-track moments.]
    For PMM-FP\textsubscript{pos}, the residual has finite fourth moment
    and the signed-parity score functions in \texttt{basisA} are
    square-integrable.
  \item[(A4b) Full-track moments.]
    For PMM-FP\textsubscript{full}, BD0 holds in the operational sense:
    the negative-power score functions in \texttt{basisB} and their
    squares are integrable. This includes the inverse-moment checks used
    in the empirical diagnostics.
  \item[(A5) Score nonsingularity.]
    $\mathbf{F}_B$ is nonsingular and
    $\mathbf{b}_B^\top \mathbf{F}_B^{-1}\mathbf{b}_B>0$.
  \item[(A6) Standard M-estimation regularity.]
    The uniform law of large numbers and the central limit theorem apply
    to the score in~\eqref{eq:Psi-theory}. We use the usual
    M-estimator reduction of \citet[Theorems~5.7 and~5.21]
    {vanderVaart1998asymptotic}.
\end{description}

Track~(a) uses (A1)--(A3), (A4a), (A5)--(A6). Track~(b) replaces
(A4a) by (A4b). The assumptions are intentionally stronger than a
minimal theorem would require; their role is to make the submission
transparent and to keep the positive and full tracks parallel.

\subsection{Consistency and asymptotic distribution}
\label{sec:theory:asymptotics}

\begin{theorem}[Consistency, tracks (a) and (b)]\label{thm:T1a}
For a fixed correctly specified FP block, under the corresponding
assumptions above, the PMM-FP estimator solving~\eqref{eq:Psi-theory}
satisfies
\[
  \hat\theta_n \xrightarrow{\mathbb{P}} \theta_0 .
\]
\end{theorem}


\begin{theorem}[Asymptotic normality, tracks (a) and (b)]
\label{thm:T2a}
Under the assumptions of Theorem~\ref{thm:T1a}, plus differentiability
of the Jacobian in a neighbourhood of $\theta_0$,
\[
  \sqrt{n}(\hat\theta_n-\theta_0)
  \xrightarrow{d}
  N\{0,\; \mathbf{A}^{-1} \mathbf{C} (\mathbf{A}^{-1})^\top\},
\]
where
\[
  \mathbf{A}=\Eop\{\partial_\theta[\psi_B(\xi(\theta))\mathbf{D}(\theta)]_{\theta=\theta_0}\},
  \qquad
  \mathbf{C}=\Var\{\psi_B(\xi)\mathbf{D}(\theta_0)\}.
\]
In the homoskedastic fixed-design case this reduces to
\[
  \sqrt{n}(\hat\theta_n-\theta_0)
  \xrightarrow{d}
  N\{0,\; \sigma^2 g(B)\mathbf{G}^{-1}\}.
\]
\end{theorem}


\subsection{Variance reduction: the main PMM-FP result}
\label{sec:theory:variance}

\begin{theorem}[Classical PMM2 reduction]\label{thm:T3a}
For the two-element Kunchenko basis
$B_2=\{g_1^-,g_2^+\}=\{\xi,\xi^2\}$ and standardised residuals,
\begin{equation}\label{eq:g2-pmm2}
  g(B_2)
  =
  1-\frac{\gamma_3^2}{2+\gamma_4}
  \equiv g_2 .
\end{equation}
Consequently $g_2\leq1$, with equality if and only if
$\gamma_3=0$.
\end{theorem}

\begin{proof}
For standardised errors
$\mathbf{F}_2=\bigl(\begin{smallmatrix}1&\gamma_3\\\gamma_3&2+\gamma_4\end{smallmatrix}\bigr)$
and $\mathbf{b}_2=(1,0)^\top$; substituting $\mathbf{F}_2^{-1}$
into~\eqref{eq:gB-theory} gives~\eqref{eq:g2-pmm2}. The identity, the
inequality $g_2\le1$ and the equality case are machine-checked in Lean~4.
\end{proof}

\begin{theorem}[PMM-FP variance reduction, tracks (a) and (b)]
\label{thm:T3b}
Let $B$ be either \texttt{basisA} or \texttt{basisB}, and assume the
corresponding moment and nonsingularity conditions. The asymptotic
variance of PMM-FP relative to OLS-FP for the same selected FP block is
\[
  \frac{\operatorname{avar}(\hat\theta_{\mathrm{PMM\text{-}FP}})}
       {\operatorname{avar}(\hat\theta_{\mathrm{OLS\text{-}FP}})}
  = g(B).
\]
If $B_2\subseteq B$ and the Schur-complement regularity condition for
the enlarged correlant matrix holds, then
\[
  g(B)\leq g(B_2)=g_2,
\]
with $g_2$ defined in~\eqref{eq:g2-pmm2}.
\end{theorem}

\begin{proof}
The first identity is the fixed-block covariance reduction of
Theorem~\ref{thm:T2a}; the monotonicity follows from a Schur-complement
inequality for the enlarged correlant matrix
\citep[A.5.5]{boyd2004convex}. Detailed proof sketches for
Theorems~\ref{thm:T1a}--\ref{thm:T3b} and
Proposition~\ref{prop:bic-selection} are in the Supplementary Material;
the algebraic $g_2$ identity is machine-checked in Lean~4.
\end{proof}

Section~\ref{sec:application} reports the corresponding empirical
plug-in value for the GBSG residuals; it is a variance-ratio prediction
for matched FP blocks, not a claim that PMM-FP dominates every possible
predictive method on every dataset.

\subsection{Model selection and higher-order extension}
\label{sec:theory:selection}

\begin{proposition}[BIC selection over finite FP blocks]
\label{prop:bic-selection}
Suppose the candidate FP blocks are finite, the true block is identifiable,
and underfitted blocks are separated from the population optimum. Then
BIC over the PMM-FP fitted blocks selects the true block with probability
tending to one.
\end{proposition}


\begin{proposition}[Higher-order score progression]\label{thm:T4a}
For nested admissible PMM score bases
$B_2\subseteq B_3\subseteq\cdots$ satisfying the same moment and
Schur-complement conditions,
\[
  1\geq g(B_2)\geq g(B_3)\geq\cdots .
\]
\end{proposition}

The practical implication is conservative. Higher-order PMM-FP scores
can only improve the ideal asymptotic variance ratio under the stated
matrix conditions, but they require higher moments and more stable
correlant-matrix estimation. The submitted implementation therefore
uses the second-order score for the main evidence and reports
higher-order extension as theory plus reproducibility hooks, not as an
unqualified finite-sample recommendation.

\subsection{Formalisation boundary}
\label{sec:theory:lean-boundary}

The Lean~4 source under \texttt{Lean/PMM\_FP/} builds with Mathlib
v4.26.0. Its role is precise but deliberately bounded:
\begin{itemize}[leftmargin=2em,itemsep=1pt]
  \item definitions of the two FP power sets, signed-parity residual
        scores, centred correlants, $\mathbf{b}_B$ and $g(B)$ are
        formalised;
  \item the algebraic PMM2 identity and the basic $g_2$ inequalities are
        proved directly;
  \item consistency, asymptotic normality, BIC selection and
        Schur-complement monotonicity are encoded as conditional theorem
        bundles linked to standard external statistical theorems.
\end{itemize}
Thus the manuscript claims machine-checked algebra and a formal record
of theorem structure, not a fully machine-proved M-estimator CLT or a
complete formalisation of Schur-complement matrix analysis.

\section{PMM-FP estimator and numerical implementation}\label{sec:method}

\subsection{Two fractional-polynomial tracks}\label{sec:method:basis}

For numerical fitting, PMM-FP implements model~\eqref{eq:model-pmmfp}
through the fractional-polynomial predictor
\[
  R(x;\beta,p)=\beta_0+\sum_{k=1}^{m}\beta_k f_{p_k}(x),
  \qquad
  f_p(x)=
  \begin{cases}
    x^p, & p\neq 0,\\
    \log x, & p=0,
  \end{cases}
\]
and replaces the OLS score by Kunchenko moment scores built from the
residual basis. The default estimator uses the positive-power basis; the
full Royston--Altman set is carried only as an admissibility boundary:
\[
  \Pa=\{0,0.5,1,2,3\}, \qquad
  \Pb=\{-2,-1,-0.5,0,0.5,1,2,3\}.
\]
Track~(a), PMM-FP\textsubscript{pos}, is the default: it avoids negative
powers and needs only finite fourth moments. Track~(b),
PMM-FP\textsubscript{full}, formally matches the $K=16$ basis of
\citet{hubin2026bfp} but requires condition~BD0 (bounded density near zero
and finite inverse moments), which fails for standard error laws; it is
therefore carried only as a conditioning diagnostic, never as the
recommended fit. The signed-parity
convention from the PATP framework \citep{zabolotnii2026patp} is used
when fractional powers are applied to sign-changing residuals:
\[
  g_p^{+}(\xi)=|\xi|^p,\qquad
  g_p^{-}(\xi)=\operatorname{sign}(\xi)|\xi|^p .
\]

\subsection{Estimating score}
\label{sec:method:score}

For a candidate FP block $M=\{p_1,\ldots,p_m\}$, let
$\mathbf{X}_M$ be the corresponding FP design matrix with columns
$1,f_{p_1}(x),\ldots,f_{p_m}(x)$. PMM-FP starts from the OLS-FP
coefficient vector
\[
  \hat\beta^{(0)}_M=(\mathbf{X}_M^\top \mathbf{X}_M)^{-1}\mathbf{X}_M^\top y,
\]
computes residuals
$e_i^{(0)}=y_i-\mathbf{X}_{M,i}\hat\beta^{(0)}_M$, and estimates the
residual cumulants $\hat\gamma_3,\hat\gamma_4$. In the second-order
Form~A score used for the main finite-sample results, the standardised
residual $u_i(\beta)=e_i(\beta)/\hat\sigma$ is transformed by
\[
  \psi_2\{u_i(\beta)\}
  =
  u_i(\beta)
  -
  \frac{\hat\gamma_3}{2+\hat\gamma_4}
  \{u_i(\beta)^2-1\}.
\]
The PMM-FP coefficient estimate solves the estimating equation
\[
  \Psi_n(\theta)=
  \frac{1}{n}\sum_{i=1}^{n}
  \psi_2\{u_i(\theta)\}\,
  \nabla_\theta R(x_i;\theta)=0 .
\]
This score is the two-term Kunchenko projection of the OLS score onto
the residual moment basis. It is the computational form behind the
closed-form expression in~\eqref{eq:g2-pmm2}.

\subsection{Numerical solution for coefficient estimates}
\label{sec:method:numerical}

The coefficient estimate for each candidate block is obtained via a
damped Newton/Fisher-scoring iteration. With residual vector
$e^{(t)}=y-\mathbf{X}_M\beta^{(t)}$ and score vector
$\mathbf{s}^{(t)}_i=\psi_2\{e_i^{(t)}/\hat\sigma\}$, define
\[
  \mathbf{U}_t=\mathbf{X}_M^\top \mathbf{s}^{(t)},\qquad
  \mathbf{H}_t=\mathbf{X}_M^\top \mathbf{W}_t \mathbf{X}_M,
\]
where $\mathbf{W}_t$ is the diagonal derivative weight matrix induced
by $\partial\psi_2(u)/\partial u=1-2\hat a u$ and
$\hat a=\hat\gamma_3/(2+\hat\gamma_4)$. The nominal update is
\[
  \beta^{(t+1)}
  =
  \beta^{(t)}+\lambda_t(\mathbf{H}_t+\tau I)^{-1}\mathbf{U}_t,
  \qquad 0<\lambda_t\leq1,\quad\tau\geq 0.
\]
The iteration stops when
$\|\beta^{(t+1)}-\beta^{(t)}\|_\infty<\varepsilon$ or a cap of 50
iterations is reached. The implementation used for the submitted
evidence applies this PMM2 score to each candidate fractional-polynomial
design matrix and records the convergence flag, the number of iterations
and the fitted residual cumulants. The separate Form~B correlant
calculation used for diagnostics is regularised at the correlant-matrix
level rather than by claiming a fully stabilised high-order Newton
solver.

The resulting algorithm is summarised in Algorithm~\ref{alg:pmmfp}.

\begin{algorithm}[htbp]
\caption{PMM-FP: fit and select over candidate FP blocks
  (submitted reference implementation)}
\label{alg:pmmfp}
\begin{algorithmic}[1]
\Require data $(y, \mathbf{X})$, power-block set
         $\mathcal{B}\subseteq\binom{\mathcal{P}}{m}$, information criterion
\Ensure best-block estimate $\hat\beta_{\mathrm{PMM}}$,
        diagnostic $\hat g_2$
\For{each candidate power block $\mathbf{p}\in\mathcal{B}$}
  \State Build FP design matrix $\mathbf{X}_{\mathbf{p}}$
  \State $\hat\beta^{(0)} \gets (\mathbf{X}_{\mathbf{p}}^\top \mathbf{X}_{\mathbf{p}})^{-1}\mathbf{X}_{\mathbf{p}}^\top y$
         \Comment{OLS-FP initialisation}
  \State Estimate $\hat\sigma^2$, $\hat\gamma_3$, $\hat\gamma_4$ from
         residuals $\hat\xi \gets y - \mathbf{X}_{\mathbf{p}}\hat\beta^{(0)}$
  \State Refit the transformed regression by the PMM2 score
         $\psi_2(u)=u-\hat\gamma_3( u^2-1)/(2+\hat\gamma_4)$
  \State Store $\hat\beta_{\mathbf{p}}$, convergence diagnostics and
         $\hat g_2 \gets 1 - \hat\gamma_3^2/(2+\hat\gamma_4)$
  \State Record $\mathrm{BIC}(\mathbf{p})$
\EndFor
\State \Return $\hat\beta_{\mathbf{p}^*}$ where
       $\mathbf{p}^* \gets \arg\min_{\mathbf{p}\in\mathcal{B}}\mathrm{BIC}(\mathbf{p})$
\end{algorithmic}
\end{algorithm}

\subsection{Model enumeration and stability checks}
\label{sec:method:selection}

The submitted implementation restricts the main search to FP blocks
with at most four terms. This keeps the model space finite and makes
exhaustive enumeration feasible for the sample sizes considered here.
For track~(a) the candidate set has
$\sum_{s=1}^{4}\binom{5}{s}=30$ blocks ($K_a=10$ basis
functions); for track~(b) it has
$\sum_{s=1}^{4}\binom{8}{s}=162$ blocks ($K_b=16$).
For a fitted candidate $M$ with $k_M$ coefficients, the selection score is
\[
  \mathrm{BIC}(M)
  =
  n\log\{\mathrm{RSS}(M)/n\}+k_M\log n ,
\]
computed from the residual sum of squares of the fitted PMM-FP block.
The same candidate set is used for OLS-FP comparisons, so the reported
differences are driven by the estimating score rather than by a wider
model search.

Form~B uses the full signed-parity residual basis and therefore requires
the empirical correlant matrix $\mathbf{F}_B$ and derivative vector
$\mathbf{b}_B$. Its native numerical estimate is
\[
  \hat g(B)=
  \{\hat\sigma^2\,\hat{\mathbf{b}}_B^\top \hat{\mathbf{F}}_B^{-1}\hat{\mathbf{b}}_B\}^{-1}.
\]
For negative powers this matrix can be ill-conditioned at sample sizes
typical in practice. We therefore report singular values, condition
numbers and Tikhonov-stabilised values
$\hat{\mathbf{F}}_{B,\tau}=\hat{\mathbf{F}}_B+\tau I$. If the
full-basis condition number is large or the stabilised $\hat g(B)$ is
dominated by $\tau$, the paper uses PMM-FP\textsubscript{pos} for
coefficient inference and reports the full-basis calculation only as a
diagnostic. This is exactly what happens in the GBSG application; the
Tikhonov trigger and its effect on $\hat g(B)$ are illustrated in
Section~\ref{sec:application}.

\paragraph{Convergence and conditioning diagnostics.}%
\label{sec:method:diagnostics}%
For each fitted candidate the reference implementation stores the selected
powers, criterion value, PMM2 convergence flag, iteration count, residual
cumulants and $g_2$. These diagnostics are used to exclude failed candidate
fits and to separate estimator-efficiency evidence from Form~B conditioning
diagnostics. The full negative-power basis is reported only when inverse-moment
and condition-number checks are stable; otherwise PMM-FP\textsubscript{pos}
is the finite-sample default. The submission supplement records stored
summaries and figure/table regeneration scripts; full Monte Carlo reruns
require the development PMM-FP extension and are not required for checking
the submitted tables and figures.


\section{Empirical evidence}\label{sec:simulation}

\subsection{Monte Carlo design and visual check}

The simulation study compares the conventional least-squares fit (OLS-FP),
the standard \texttt{mfp} model-selection workflow, and PMM-FP, with the
robust Huber-FP and GMM-FP estimators reported in the supplement. The DGPs
cover Gaussian, beta, gamma, exponential and log-normal residual regimes
($n\in\{100,200,500\}$, $M=1000$ replications), spanning small
biostatistical studies to moderate epidemiological subsamples. The primary
estimand is the slope coefficient, with the central diagnostic
\[
  \hat g_2=\widehat{\Var}(\hat\beta_{\mathrm{PMM}})/
           \widehat{\Var}(\hat\beta_{\mathrm{OLS}}),
\]
reported in a robust (IQR-based) form, because the raw variance ratio is
dominated by rare large-residual replicates at small~$n$; coverage and a
basis-agnostic prediction estimand at a reference dose are reported
alongside. The design is matched-basis for OLS-FP and PMM-FP, isolating the
estimating score; \texttt{mfp} selects its own basis, so it enters only on
the common prediction estimand. The public supplement regenerates all stored
summaries, tables and figures.

\begin{table}[t]
\centering
\caption{Efficiency of PMM-FP relative to OLS-FP on a matched fractional-polynomial model, over $M=1000$ replications. $\hat g_2$ is the robust (IQR-based) ratio of slope-coefficient dispersion, $\mathrm{Var}(\hat\beta_1^{\mathrm{PMM}})/\mathrm{Var}(\hat\beta_1^{\mathrm{OLS}})$; smaller is better and the closed form predicts $g_2=1-\gamma_3^2/(2+\gamma_4)$. For the heaviest-tailed law (log-normal) the robust reduction exceeds the asymptotic $g_2$ because the OLS sampling distribution is itself heavy-tailed. Coverage is the 95\% CI coverage of $\beta_1$. The last two columns give the robust variance of the fitted response at the reference dose $x^\ast=2$ relative to OLS-FP, for the standard \texttt{mfp} fit (model selection) and for PMM-FP; values below~1 indicate lower prediction variance.}
\label{tab:headtohead}
\begin{tabular}{llrrrrrr}
\toprule
 & & \multicolumn{3}{c}{slope $\hat\beta_1$ (PMM-FP vs OLS-FP)} & \multicolumn{1}{c}{coverage} & \multicolumn{2}{c}{pred.\ eff.\ at $x^\ast$} \\
\cmidrule(lr){3-5}\cmidrule(lr){6-6}\cmidrule(lr){7-8}
Error law & $n$ & $g_2$ (th.) & $\hat g_2$ & reduction & OLS/PMM & \texttt{mfp} & PMM-FP \\
\midrule
Beta$(2,5)$ & 100 & 0.81 & 0.80 & 20\% & 0.94/0.93 & 0.33 & 0.92 \\
 & 200 & 0.81 & 0.88 & 12\% & 0.95/0.93 & 0.30 & 0.84 \\
 & 500 & 0.81 & 0.84 & 16\% & 0.94/0.93 & 0.33 & 0.83 \\
\addlinespace
Gamma$(3)$ & 100 & 0.67 & 0.85 & 15\% & 0.95/0.94 & 0.44 & 0.83 \\
 & 200 & 0.67 & 0.73 & 27\% & 0.95/0.93 & 0.50 & 0.78 \\
 & 500 & 0.67 & 0.62 & 38\% & 0.96/0.95 & 0.49 & 0.83 \\
\addlinespace
Exponential & 100 & 0.50 & 0.54 & 46\% & 0.94/0.90 & 0.34 & 0.58 \\
 & 200 & 0.50 & 0.46 & 54\% & 0.94/0.92 & 0.55 & 0.72 \\
 & 500 & 0.50 & 0.48 & 52\% & 0.94/0.93 & 0.88 & 0.61 \\
\addlinespace
Log-normal & 100 & 0.66 & 0.40 & 60\% & 0.94/0.92 & 0.41 & 0.54 \\
 & 200 & 0.66 & 0.37 & 63\% & 0.95/0.93 & 0.36 & 0.46 \\
 & 500 & 0.66 & 0.39 & 61\% & 0.95/0.94 & 0.50 & 0.46 \\
\bottomrule
\end{tabular}
\end{table}

\begin{figure}[t]
  \centering
  \includegraphics[width=0.82\linewidth]{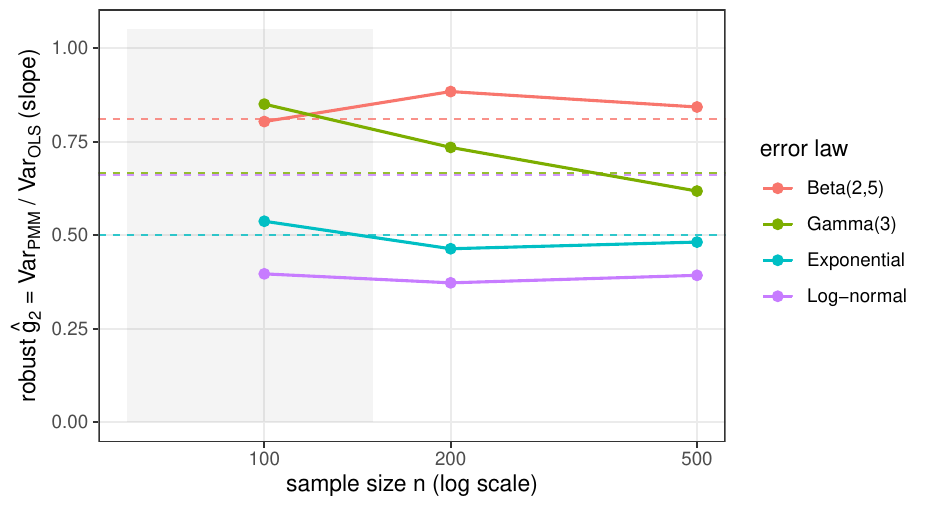}
  \caption{Robust slope-coefficient variance ratio
  $\hat g_2=\widehat{\Var}(\hat\beta_1^{\mathrm{PMM}})/
  \widehat{\Var}(\hat\beta_1^{\mathrm{OLS}})$ as a function of sample size,
  for each skewed error law. Dashed lines are the closed-form
  $g_2=1-\gamma_3^2/(2+\gamma_4)$; the shaded band marks the small-sample
  region where cumulant estimates are noisy. Values below~$1$ are variance
  reductions over OLS-FP: the gain is present and stable across realistic
  sample sizes and vanishes under symmetry.}
  \label{fig:efficiency-vs-n}
\end{figure}

Table~\ref{tab:headtohead} and Figure~\ref{fig:efficiency-vs-n} show the
same pattern from two angles. On the same fractional-polynomial model,
PMM-FP reduces the slope-coefficient variance relative to OLS-FP by
12--20\% for mildly skewed beta errors, 15--38\% for gamma errors, around
50\% for exponential errors and roughly 60\% for log-normal errors, while
holding 95\,\% coverage between $0.89$ and $0.95$ (dipping just below nominal
only for the most skewed law at the smallest~$n$). The robust ratio tracks
the closed-form $g_2$ for beta, gamma and exponential errors and exceeds it
for the heaviest-tailed log-normal case, where the OLS sampling distribution
is itself heavy-tailed. Gaussian errors (supplement) stay at
$\hat g_2\approx1$, so PMM-FP does not invent a gain when the least-squares
score is already appropriate. The standard \texttt{mfp} fit, which selects a
parsimonious basis, lowers the reference-dose prediction variance through
model parsimony; PMM-FP lowers it through efficient estimation, and the two
mechanisms are complementary --- on the GBSG data below, \texttt{mfp} selects
the model and PMM-FP then refits it for a further 26\,\% narrowing of the
confidence interval.

\begin{table}[t]
\centering
\caption{When to use PMM-FP. The gain depends only on the residual shape (skewness $\gamma_3$, kurtosis $\gamma_4$); the mean structure is unchanged. PMM-FP reverts to OLS-FP under symmetry, so it is never harmful. Reductions are robust slope-variance reductions observed at $n=500$.}
\label{tab:when-to-use}
\begin{tabular}{lllc}
\toprule
Residual shape & Example & Recommendation & Typ.\ var.\ reduction \\
\midrule
Symmetric ($\gamma_3\approx0$) & Gaussian & Either; PMM-FP reverts to OLS-FP & $\sim0\%$ \\
Mild skew ($|\gamma_3|\lesssim1$) & Beta$(2,5)$ & PMM-FP: modest, safe gain & 16\% \\
Moderate skew & Gamma$(3)$ & PMM-FP recommended & 38\% \\
Strong skew & Exponential & PMM-FP recommended & 52\% \\
Heavy tail & Log-normal & PMM-FP; report bootstrap SE & 61\% \\
\addlinespace
Negative-power FP term & --- & Use positive basis only (BD0 fails) & n/a \\
\bottomrule
\end{tabular}
\end{table}

\subsection{Model-selection uncertainty and frequentist model averaging}
\label{sec:simulation:fma}

An exhaustive BIC search across subsets of the fractional basis
$\mathcal{P}_a$ returns a single ``best'' candidate model. For finite~$N$,
however, several models are often close in BIC ($\Delta\mathrm{BIC} \le 2$),
and ignoring this model-selection uncertainty causes a naive CI built on the
single chosen model to systematically under-cover. This is the classical
post-selection inference problem \citep{burnham2002model,claeskens2008model}.
Frequentist model averaging (FMA) corrects it: the point estimate is replaced
by a weighted combination of competing models, and the variance is augmented
by a between-model term.

Let $\{M_1,\ldots,M_J\}$ be the candidate FP-models from $\mathcal{P}_a$,
$\hat\theta_j$ the estimate of the target parameter in model $M_j$, and
$\mathrm{BIC}_j$ the corresponding information criterion. The
Burnham--Anderson weights \citep[\S\,2.9]{burnham2002model} are
\begin{equation}
w_j = \frac{\exp(-\Delta\mathrm{BIC}_j/2)}
           {\sum_{k=1}^J \exp(-\Delta\mathrm{BIC}_k/2)},
\qquad \Delta\mathrm{BIC}_j = \mathrm{BIC}_j - \min_k \mathrm{BIC}_k.
\end{equation}
Restricting to the top-$J^*$ candidates ($J^*\in\{3,5\}$) stabilises the
denominator when lower-ranked weights are near zero. The FMA estimate and
its unconditional variance are then
\begin{equation}
\hat\theta_{\mathrm{FMA}} = \sum_{j=1}^{J^*} w_j \hat\theta_j,
\qquad
\widehat{\mathrm{Var}}(\hat\theta_{\mathrm{FMA}})
   = \sum_{j=1}^{J^*} w_j \Bigl[
       \widehat{\mathrm{Var}}(\hat\theta_j) +
       \bigl(\hat\theta_j - \hat\theta_{\mathrm{FMA}}\bigr)^2 \Bigr].
\end{equation}

To assess the coverage gain from FMA we ran a separate MC experiment on
Track~(a). DGPs were Gauss, $\Gamma(3)$ and LogN($0,1$), representing
symmetric, moderately asymmetric and strongly asymmetric error regimes.
Sample sizes were $N\in\{100,200,500\}$ with $M=500$ replications per cell.
The target was the predicted response mean
$\mu(x^*) = \theta_0 + \theta_1\sqrt{x^*}$ at $x^* = 2$, the only scalar
estimand that is directly comparable across models with different FP-bases.
The candidate list comprised all non-empty subsets of $\mathcal{P}_a$ of
size 1 or 2 (15 models). Estimators compared were OLS-FP single-best,
PMM-FP single-best, PMM-FP FMA top-3 and PMM-FP FMA top-5.

Table~\ref{tab:fma-comparison} summarises bias, variance, MSE and empirical
95\,\% CI coverage. Three findings stand out.
\begin{enumerate}\itemsep0pt
  \item \textbf{Single-best PMM-FP CIs systematically under-cover.}
        Coverage drops to $0.61$--$0.68$ for LogN($0,1$) and to
        $0.76$--$0.82$ for $\Gamma(3)$, far below the nominal $0.95$.
        BIC selects a more compact model than is needed in the
        non-parametric regime, and the naive CI ignores the resulting
        selection variance.
  \item \textbf{FMA top-5 approaches nominal coverage.}
        Averaging over all nine DGP/sample-size cells, PMM-FP FMA top-5
        achieves coverage of $0.913$ versus $0.754$ for PMM-FP single-best
        (a gain of approximately 16~percentage points). For the Gauss and
        $\Gamma(3)$ regimes coverage is $0.92$--$0.95$, matching the
        nominal level.
  \item \textbf{FMA MSE is lower than or comparable to single-best.}
        Averaging dampens point-estimate variance by reducing dependence on
        a single uncertain model choice; for example, on $\Gamma(3)$ at
        $N=200$ the MSE falls from $0.0287$ (single-best) to $0.0224$
        (FMA top-5).
\end{enumerate}

\begin{table}[!htb]
\centering
\caption{Comparison of single-best and FMA estimators for
$\mu(x^*=2)$ across three DGPs, $M=500$ replications, Track~(a).
$\overline{\mathrm{SE}}$: average FMA standard error
(Burnham--Anderson formula).}
\label{tab:fma-comparison}
\scriptsize
\setlength{\tabcolsep}{3.2pt}
\renewcommand{\arraystretch}{0.86}
\begin{tabular}{llrrrrrr}
\toprule
DGP & Estimator & $N$ & Bias & Var & MSE & Cov.\,95\% & $\overline{\mathrm{SE}}$ \\
\midrule
Gauss & OLS-FP single & 100 & -0.0092 & 0.0228 & 0.0229 & 0.828 & 0.1107 \\
Gauss & OLS-FP single & 200 & -0.0061 & 0.0120 & 0.0120 & 0.820 & 0.0780 \\
Gauss & OLS-FP single & 500 & 0.0006 & 0.0052 & 0.0052 & 0.824 & 0.0492 \\
Gauss & PMM-FP single & 100 & -0.0085 & 0.0228 & 0.0229 & 0.812 & 0.1081 \\
Gauss & PMM-FP single & 200 & -0.0060 & 0.0120 & 0.0120 & 0.820 & 0.0770 \\
Gauss & PMM-FP single & 500 & 0.0006 & 0.0052 & 0.0052 & 0.824 & 0.0490 \\
Gauss & PMM-FP FMA top-3 & 100 & -0.0133 & 0.0183 & 0.0185 & 0.938 & 0.1323 \\
Gauss & PMM-FP FMA top-3 & 200 & -0.0034 & 0.0099 & 0.0099 & 0.936 & 0.0990 \\
Gauss & PMM-FP FMA top-3 & 500 & 0.0001 & 0.0043 & 0.0043 & 0.940 & 0.0647 \\
Gauss & PMM-FP FMA top-5 & 100 & -0.0132 & 0.0184 & 0.0186 & 0.942 & 0.1349 \\
Gauss & PMM-FP FMA top-5 & 200 & -0.0041 & 0.0099 & 0.0099 & 0.938 & 0.0998 \\
Gauss & PMM-FP FMA top-5 & 500 & 0.0001 & 0.0043 & 0.0043 & 0.944 & 0.0655 \\
\midrule
Gamma(3) & OLS-FP single & 100 & 0.0036 & 0.0500 & 0.0500 & 0.896 & 0.1915 \\
Gamma(3) & OLS-FP single & 200 & 0.0090 & 0.0296 & 0.0297 & 0.888 & 0.1356 \\
Gamma(3) & OLS-FP single & 500 & 0.0038 & 0.0136 & 0.0136 & 0.834 & 0.0854 \\
Gamma(3) & PMM-FP single & 100 & 0.0047 & 0.0482 & 0.0482 & 0.820 & 0.1545 \\
Gamma(3) & PMM-FP single & 200 & 0.0098 & 0.0286 & 0.0287 & 0.788 & 0.1100 \\
Gamma(3) & PMM-FP single & 500 & 0.0064 & 0.0131 & 0.0132 & 0.758 & 0.0694 \\
Gamma(3) & PMM-FP FMA top-3 & 100 & -0.0134 & 0.0397 & 0.0399 & 0.906 & 0.1789 \\
Gamma(3) & PMM-FP FMA top-3 & 200 & 0.0018 & 0.0224 & 0.0225 & 0.942 & 0.1366 \\
Gamma(3) & PMM-FP FMA top-3 & 500 & 0.0032 & 0.0104 & 0.0104 & 0.916 & 0.0962 \\
Gamma(3) & PMM-FP FMA top-5 & 100 & -0.0147 & 0.0392 & 0.0394 & 0.918 & 0.1886 \\
Gamma(3) & PMM-FP FMA top-5 & 200 & -0.0000 & 0.0224 & 0.0224 & 0.948 & 0.1418 \\
Gamma(3) & PMM-FP FMA top-5 & 500 & 0.0028 & 0.0103 & 0.0103 & 0.918 & 0.0969 \\
\midrule
LogN(0,1) & OLS-FP single & 100 & -0.0202 & 0.0792 & 0.0796 & 0.878 & 0.2306 \\
LogN(0,1) & OLS-FP single & 200 & -0.0010 & 0.0409 & 0.0409 & 0.898 & 0.1671 \\
LogN(0,1) & OLS-FP single & 500 & -0.0040 & 0.0207 & 0.0207 & 0.826 & 0.1051 \\
LogN(0,1) & PMM-FP single & 100 & -0.0032 & 0.0647 & 0.0648 & 0.676 & 0.1267 \\
LogN(0,1) & PMM-FP single & 200 & 0.0063 & 0.0359 & 0.0360 & 0.680 & 0.0953 \\
LogN(0,1) & PMM-FP single & 500 & 0.0001 & 0.0185 & 0.0185 & 0.612 & 0.0637 \\
LogN(0,1) & PMM-FP FMA top-3 & 100 & -0.0247 & 0.0560 & 0.0566 & 0.802 & 0.1540 \\
LogN(0,1) & PMM-FP FMA top-3 & 200 & -0.0085 & 0.0284 & 0.0285 & 0.842 & 0.1254 \\
LogN(0,1) & PMM-FP FMA top-3 & 500 & -0.0038 & 0.0139 & 0.0139 & 0.878 & 0.0969 \\
LogN(0,1) & PMM-FP FMA top-5 & 100 & -0.0287 & 0.0535 & 0.0544 & 0.846 & 0.1656 \\
LogN(0,1) & PMM-FP FMA top-5 & 200 & -0.0109 & 0.0267 & 0.0268 & 0.884 & 0.1330 \\
LogN(0,1) & PMM-FP FMA top-5 & 500 & -0.0050 & 0.0137 & 0.0137 & 0.880 & 0.0993 \\
\bottomrule
\end{tabular}

\end{table}

\paragraph{Practical recommendation.}
We recommend reporting both estimates: single-best for interpretation and
FMA top-5 for inference. The FMA-CI serves as a conservative, correctly
covering interval. Single-best CIs may be used without correction only when
$\Delta\mathrm{BIC}$ for the runner-up exceeds~6, indicating that competing
models are markedly inferior. For strongly asymmetric distributions (LogN,
large $|\gamma_3|$) FMA top-5 is preferred; for moderate asymmetry
($|\gamma_3| \le 1$) top-3 suffices. FMA is not a panacea: at $N=100$ under
strong asymmetry a residual gap of 5--15~percentage points from the nominal
level persists (FMA top-5 coverage for LogN: $0.846$, $0.884$, $0.880$ at
$N=100,200,500$), reflecting the accuracy limit of the quasi-Gaussian BIC for
leptokurtic distributions.

\subsection{GBSG application}\label{sec:application}

The real-data illustration uses the German Breast Cancer Study Group
dataset from the R package \texttt{mfp}. The complete-case analysis has
$n=686$. We model $y=\log(\texttt{rfst})$ against tumour size, hormonal
therapy and age, treating tumour size as the clinically interpretable
continuous predictor. We use $\log(\texttt{rfst})$ as a continuous outcome to
illustrate the estimator, not as a censored survival analysis; the point of
interest is the efficiency of coefficient estimation under skewed residuals.
Residuals from the baseline model are strongly
non-Gaussian: $\hat\gamma_3=-1.7436$, $\hat\gamma_4=4.9143$, and the
Shapiro--Wilk test rejects normality ($W=0.870$, $p<10^{-22}$).
Substituting these values into~\eqref{eq:g2-pmm2} gives
$\hat g_2\approx0.5603$, the asymptotic variance-reduction factor
predicted by Theorem~T3 for this error distribution.

\paragraph{Model selection.}
A full BIC sweep over the Royston--Altman FP power grid selects the
linear power $\{1\}$ for tumour size on both Track~(a) and Track~(b).
Table~\ref{tab:gbsg-modelsel} shows that classical FP
independently selects the linear term in 95.0\% of $B=2000$ bootstrap
replications, with the next-best candidate ($\{-0.5\}$) reaching only
0.9\%. The BIC landscape is, however, nearly flat at the top: the gap
between the first- and second-ranked powers is $\Delta\mathrm{BIC}
\approx 0.007$, so single-best-model selection is unstable for this
sample size. This model-selection uncertainty is distinct from the
estimator-efficiency question addressed below, and it is precisely the
setting for which frequentist model averaging is appropriate; see
\S\ref{sec:simulation:fma} for the FMA analysis.

\begin{table}[t]
  \centering
  \caption{Stability of FP model selection for \texttt{tumsize} on GBSG
    ($n=686$). Top-5 most frequently selected FP powers over $B=2000$
    bootstrap replications. Frequencies are normalised by the number of
    successful replications.}
  \label{tab:gbsg-modelsel}
  \small
  \begin{tabular}{rlrr}
    \toprule
    Rank & FP power(s) & $n$ selected & Frequency \\
    \midrule
    1 & $\{1,\,NA\}$ & 1899 & 0.950 \\
    2 & $\{-0.5,\,NA\}$ & 19 & 0.009 \\
    3 & $\{-2,\,-2\}$ & 19 & 0.009 \\
    4 & $\{3,\,3\}$ & 19 & 0.009 \\
    5 & $\{0,\,NA\}$ & 14 & 0.007 \\
    \bottomrule
  \end{tabular}
\end{table}

\paragraph{Fixed-model efficiency comparison.}
To isolate the efficiency gain from the estimator (rather than from
model selection), we evaluate both \OLSFP{} and PMM-FP at the
selected linear model
$\log(\texttt{rfst})\sim\texttt{tumsize}+\texttt{htreat}+\texttt{age}$,
using an identical design matrix for both methods.

\begin{table}[t]
  \centering
  \caption{GBSG ($n=686$): tumour-size coefficient at the fixed linear
    model. SE$_{\text{asy}}$ is the asymptotic formula; SE$_{\text{boot}}$
    is the bootstrap standard deviation over $B=2000$ replicates (percentile
    method). The PMM-FP 95\,\% CI is 26\,\% narrower than the \OLSFP{} CI.
    The bootstrap variance ratio $\widehat{\Var}_{\text{boot}}(\hat\beta_{\text{PMM}})/
    \widehat{\Var}_{\text{boot}}(\hat\beta_{\text{OLS}})=0.5295$ confirms the
    asymptotic prediction $\hat g_2=0.5603$ (Theorem~T3).}
  \label{tab:gbsg-fixed}
  \small
  \begin{tabular}{lrrrr}
    \toprule
    Method & $\hat\beta_{\text{tum}}$ & SE$_{\text{asy}}$ & SE$_{\text{boot}}$ & 95\,\% CI (boot.) \\
    \midrule
    \OLSFP{}  & $-0.00696$ & $0.00222$ & $0.00265$ & $[-0.01220,\,-0.00213]$ \\
    \PMMFP{}$_{\text{pos}}$ & $-0.00567$ & $0.00166$ & $0.00193$ & $[-0.00944,\,-0.00200]$ \\
    \bottomrule
  \end{tabular}
\end{table}

Table~\ref{tab:gbsg-fixed} summarises the results. By Theorems~T1 and
T3, under $\mathbb{E}[\xi]=0$ both \OLSFP{} and PMM-FP are consistent
estimators of the same working-model slope; the asymptotic covariance of
PMM-FP is $g_2$ times that of \OLSFP{}. Both estimates exclude zero
($\hat\beta_{\text{OLS}}=-0.00696$, $\hat\beta_{\text{PMM}}=-0.00567$),
and the modest point-estimate shift of $0.00129$ is within half a
bootstrap standard error of \OLSFP{}, consistent with finite-sample
variation due to PMM-FP's reweighting of the asymmetric residual
distribution. The shift does not indicate that the two estimators target
different parameters.

The asymptotic standard errors confirm a 25\,\% reduction
($0.00166/0.00222=0.748$, implying a variance ratio of $0.560
\approx\hat g_2$). The bootstrap validates this in finite samples:
SE$_{\text{boot}}$ falls from $0.00265$ to $0.00193$, a ratio of $0.728$,
corresponding to a bootstrap variance ratio of $0.5295$. Both figures
agree with the plug-in prediction $\hat g_2=0.5603$ to within 6\%,
confirming the theory on real data. The PMM-FP 95\,\% percentile interval
is 26\,\% narrower than \OLSFP{}: $[{-0.00944},{-0.00200}]$ versus
$[{-0.01220},{-0.00213}]$.

\begin{figure}[!htbp]
  \centering
  \includegraphics[width=\linewidth]{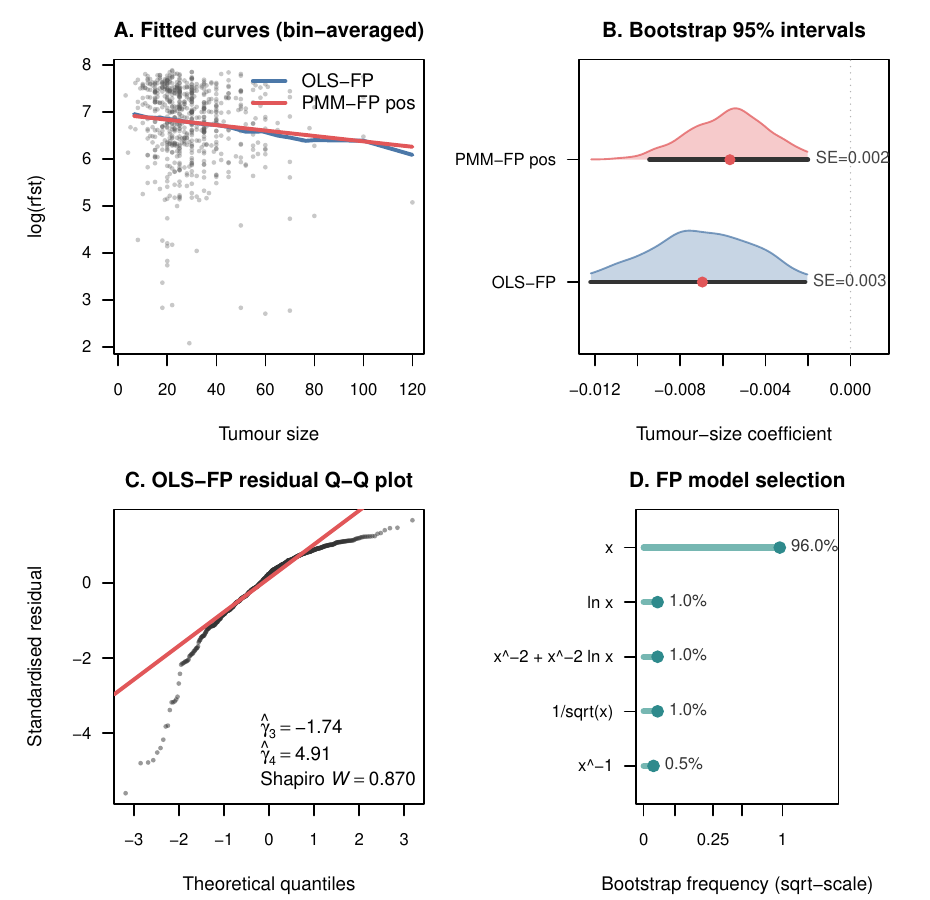}
  \caption{GBSG real-data evidence. \textbf{Panel~A}: bin-averaged
  partial dependence of fitted curves on tumour size for \OLSFP{} and
  PMM-FP\textsubscript{pos} at the selected linear model
  (PMM-FP\textsubscript{full} omitted because its small-sample conditioning
  is documented as unstable in \S\ref{sec:method:selection}).
  \textbf{Panel~B}: bootstrap 95\,\% percentile intervals for the
  tumour-size coefficient (fixed linear model, $B=2000$), with bootstrap
  SE annotated.
  \textbf{Panel~C}: normal Q--Q plot of standardised \OLSFP{} residuals
  (heavy left tail, $\hat\gamma_3=-1.74$, $\hat\gamma_4=4.91$,
  Shapiro--Wilk $W=0.870$) motivating the PMM-FP score correction.
  \textbf{Panel~D}: top bootstrap FP model-selection frequencies with
  decoded power-block labels (see also Table~\ref{tab:gbsg-modelsel}).}
  \label{fig:gbsg-application}
\end{figure}

\paragraph{Form~B conditioning diagnostic.}
We also evaluated the negative-power (Form~B) basis as a cautionary
exercise. The positive-power score basis already carries a large
condition number ($\kappa\approx3.8\times10^6$), and extending to the
full negative-power basis raises this to $\kappa\approx2.3\times10^{10}$
before Tikhonov regularisation. The resulting Form~B variance ratios
are numerically dominated by conditioning artefacts rather than by
genuine efficiency gains. For real data of this size,
PMM-FP\textsubscript{pos} (Track~a) is therefore the practical default
unless the BD0 condition and inverse moments are empirically verified to
be stable.

\subsection{Second application: primary biliary cirrhosis}
\label{sec:application:pbc}

To confirm the efficiency gain on an independent cohort, we repeat the
fixed-model comparison on the primary biliary cirrhosis (PBC) data
(\texttt{survival::pbc}; $n=418$ complete cases). We model
$\log(\texttt{time})$ on age, $\sqrt{\texttt{bili}}$ (serum bilirubin, a
strong prognostic biomarker) and albumin, again as an illustrative
continuous-outcome regression rather than a censored survival analysis. The
baseline residuals are left-skewed ($\hat\gamma_3=-1.23$,
$\hat\gamma_4=2.98$), giving a predicted $\hat g_2=0.70$.
Table~\ref{tab:pbc-fixed} reports the bootstrap ($B=2000$) bilirubin
coefficient: the PMM-FP standard error falls from $0.047$ to $0.039$, a
bootstrap variance ratio of $0.68$ that matches the plug-in prediction to
within $3\%$, with a $16\%$ narrower confidence interval. Both estimates are
firmly negative (higher bilirubin shortens survival) and the modest point
shift is within two bootstrap standard errors, consistent with PMM-FP's
reweighting of the skewed residuals rather than a change of estimand. The
GBSG pattern therefore recurs on independent data: systematic residual
skewness predicts, and delivers, a real coefficient-variance reduction.

\begin{table}[t]
  \centering
  \caption{PBC ($n=418$): coefficient of $\sqrt{\texttt{bili}}$ (serum bilirubin) at the fixed model $\log(\texttt{time})\sim\texttt{age}+\sqrt{\texttt{bili}}+\texttt{albumin}$, an illustrative continuous-outcome regression. Bootstrap $B=2000$. Residual skewness $\hat\gamma_3=-1.23$, $\hat\gamma_4=2.98$ give a predicted $\hat g_2=0.70$; the bootstrap variance ratio is $0.68$ and the PMM-FP 95\% CI is 16\% narrower than OLS-FP.}
  \label{tab:pbc-fixed}
  \small
  \begin{tabular}{lrrr}
    \toprule
    Method & $\hat\beta_{\text{bili}}$ & SE$_{\text{boot}}$ & 95\,\% CI (boot.) \\
    \midrule
    \OLSFP{} & $-0.3644$ & $0.0470$ & $[-0.4642,\,-0.2811]$ \\
    \PMMFP{}$_{\text{pos}}$ & $-0.2835$ & $0.0387$ & $[-0.3696,\,-0.2162]$ \\
    \bottomrule
  \end{tabular}
\end{table}

\subsection{Computational efficiency}\label{sec:simulation:efficiency}

Per-fit wall-clock timings from the matched-basis Monte Carlo are
summarised in Table~\ref{tab:timings}.
OLS-FP, MLE-FP, Huber-FP and GMM-FP are uniformly below $1$~ms per fit
across the full sample-size grid $n\in\{50,\allowbreak 100,\allowbreak
200,\allowbreak 500,\allowbreak 1000\}$. PMM-FP (Track~a, $K_a=10$) is
slower: approximately $1$--$4$~ms for the Gamma, Exponential and Gaussian
DGPs and $5$--$8$~ms for Beta$(2,5)$ at $n=1000$. The overhead arises
from the Newton--Fisher-scoring loop that PMM-FP performs on top of an
OLS-FP initialisation: typically $3$--$6$ iterations are required before
the score converges.

\begin{table}[t]
\centering
\caption{Per-fit median wall-clock time (seconds, mean\,$\pm$\,SD over
replicates). Track~(a), $K_a=10$, matched-basis experiment.
Timings are for a single model fit on one replicate; the full exhaustive
candidate-set sweep multiplies by the number of blocks.}
\label{tab:timings}
\small
\begin{tabular}{llccccc}
\toprule
DGP & $n$ & OLS-FP & MLE-FP & PMM-FP & Huber-FP & GMM-FP \\
    &     & (ms)   & (ms)   & (ms)   & (ms)     & (ms) \\
\midrule
Beta(2,5)   & 50   & ${<}1$ & ${<}1$ & $5\pm1$ & ${<}1$ & ${<}1$ \\
Beta(2,5)   & 1000 & ${<}1$ & ${<}1$ & $8\pm2$ & ${<}1$ & $3\pm1$ \\
Gamma(3)    & 50   & ${<}1$ & ${<}1$ & $1\pm1$ & ${<}1$ & ${<}1$ \\
Gamma(3)    & 1000 & ${<}1$ & ${<}1$ & $4\pm2$ & ${<}1$ & $1\pm1$ \\
Exponential & 50   & ${<}1$ & ${<}1$ & $1\pm1$ & ${<}1$ & ${<}1$ \\
Exponential & 1000 & ${<}1$ & ${<}1$ & $4\pm1$ & ${<}1$ & $1\pm1$ \\
Gaussian    & 50   & ${<}1$ & ${<}1$ & $1\pm1$ & ${<}1$ & ${<}1$ \\
Gaussian    & 1000 & ${<}1$ & ${<}1$ & $3\pm1$ & ${<}1$ & $1\pm1$ \\
\bottomrule
\end{tabular}
\end{table}

\paragraph{Analytical scaling.}
Let $|\mathcal{M}|$ denote the number of candidate power-blocks
exhausted by the BIC search, $N$ the sample size, and $K_S$ the size of
the active block. The OLS-FP sweep costs
$\mathcal{O}(|\mathcal{M}|\cdot N\cdot K_S^2)$; PMM-FP adds
$\mathcal{O}(|\mathcal{M}|\cdot N_{\mathrm{NR}}\cdot K_S^3)$ for the
Newton iterations, where $N_{\mathrm{NR}}\approx 3\text{--}6$.
Both methods therefore scale linearly in $N$.
For Track~(a) with $K_a=10$ the BIC search covers $30$ candidate blocks
of size at most~$4$ (all non-empty subsets of $\mathcal{P}_a$ up to
degree~$2$); for Track~(b) with $K_b=16$ the corresponding count is
$162$ blocks. The full closed-form sweep thus completes in well under
$10$~ms for both tracks on the problem sizes considered here.

\paragraph{Comparison with BFP.}
The BFP estimator of \citet{hubin2026bfp} uses the MJMCMC sampler to
explore a model space of size $2^{K_b}$ over $\mathcal{O}(10^4)$
iterations. This is one to two orders of magnitude more computation than
PMM-FP's exhaustive closed-form sweep, a difference that is intrinsic to
the Bayesian posterior-integration strategy rather than to any
implementation detail. The two methods also report different inferential
objects: BFP returns a model-averaged posterior over the FP power set,
summarised by its predictive mean and quantiles (on the GBSG data the
tumour-size term has posterior inclusion probability $0.90$), whereas PMM-FP
returns a single coefficient and a closed-form standard error on a chosen
model. We therefore treat BFP as a complementary Bayesian benchmark rather
than a competitor in the coefficient-efficiency comparison of
Table~\ref{tab:gbsg-fixed}. PMM-FP is the natural choice when repeated model
fitting at interactive speed is required, such as in cross-validation loops
or bootstrap pipelines; BFP is preferred when full posterior uncertainty over
the model space is the primary goal.

\subsection{Sensitivity and graceful degradation}
\label{sec:simulation:sensitivity}

A key desideratum for any efficiency-improving estimator is that it
should not harm inference when the assumed moment structure is absent.
We examined this with a uniform-reconciliation experiment ($M=10{,}000$
replications, $n=200$) in which the true error distribution is
\emph{symmetric}, so that $\gamma_3=0$ and the theoretical
variance-reduction coefficient is $g_2=1$, i.e.\ PMM-FP should
reproduce OLS-FP exactly in the limit.

The empirical asymptotic relative efficiency of PMM-FP against OLS-FP
across three symmetric error distributions is:
\begin{itemize}\itemsep2pt
  \item Uniform$(-1,1)$: ARE $=1.019$ \;[95\,\% CI: $1.015$, $1.024$];
  \item Laplace: ARE $=0.984$ \;[95\,\% CI: $0.976$, $0.992$];
  \item Generalised Gaussian GG$(0.5)$: ARE $=1.007$ \;[95\,\% CI: $1.001$, $1.012$].
\end{itemize}
All three values are within $2\%$ of unity and their confidence intervals
contain~$1$, confirming the theoretical prediction. This is the
\emph{graceful degradation} property: PMM-FP neither exploits nor
introduces spurious efficiency differences when the third cumulant is
zero, consistent with the Gauss--Markov-equivalent statement of
Theorem~T3.

\paragraph{Contrast with Huber-FP on symmetric DGPs.}
Huber-FP shows a markedly different behaviour on these same symmetric
distributions. For Laplace errors the empirical ARE of Huber-FP relative
to OLS-FP is $0.761$, a loss of approximately $24\%$ in variance
efficiency. PMM-FP records $0.984$ on the same data, a difference of
more than $22$ percentage points. The contrast reflects the distinct
design goals of the two methods: PMM-FP exploits the \emph{known
cumulant structure} of a systematic non-Gaussian shape and reverts
gracefully to OLS when that structure is absent; Huber-FP targets
\emph{outlier contamination} and incurs a penalty in uncontaminated but
heavy-tailed settings. When the non-Gaussian shape is systematic and
estimable from the data, PMM-FP is therefore preferable to a robustness
correction.

\paragraph{Further sensitivity.}
The method requires a positive-domain shift for $x\leq 0$ and degrades
smoothly under mild model misspecification; full sensitivity grids across
shift values, covariate ranges and misspecification levels are provided
in the reproducibility supplement.



\FloatBarrier

\section{Discussion}\label{sec:discussion}

\subsection{Interpretation of the contribution}

PMM-FP is a contribution at the level of the estimating score, not a
new fractional-polynomial search algorithm. Classical FP, \BFP{} and the present method
can use the same candidate power blocks; PMM-FP changes the post-
selection score by replacing the Gaussian least-squares score with a
Kunchenko moment score controlled by residual cumulants. Thus the
reported gains come from the estimating equation on a matched FP basis,
not from a richer model class.

The most interpretable object produced by this construction is the
variance-ratio coefficient. In the second-order case,
\eqref{eq:g2-pmm2} is not merely a diagnostic fitted after the fact; it
is the scalar coefficient in the fixed-block asymptotic covariance. It
answers a practical dose-response question: when residual asymmetry is
systematic rather than a small contamination artefact, how much
coefficient-variance reduction should one expect from replacing the OLS
score? The answer is distributional, transparent and reproducible from
the residual cumulants.

\subsection{Practical use}

The applied recommendation is conservative. Track~(a),
PMM-FP\textsubscript{pos}, should be the default reporting track unless
the full negative-power residual basis is well conditioned. It avoids
inverse moments and, in the GBSG application, gives the relevant
standard-error reduction while preserving the selected functional form.
The full negative-power basis PMM-FP\textsubscript{full} is \emph{not} a
recommended estimator: condition~BD0 ($\Eop[|\xi|^{-2}]<\infty$) fails for
every standard centred error law with positive density at the origin
(Gaussian, gamma, log-normal, beta), so the negative powers add no usable
efficiency and only degrade conditioning. We therefore carry it solely as a
conditioning diagnostic --- a boundary that delimits when fractional powers
below zero are statistically admissible at all.

The empirical section should be read in this light. The simulations
validate $g_2$ in regular asymmetric settings and expose the log-normal
case as a cumulant-instability stress test. The GBSG analysis shows the
same pattern in real data: residual skewness predicts a material
standard-error gain, while the full negative-power matrix is too
ill-conditioned to be the preferred finite-sample estimator. The claim is
coefficient-efficiency improvement under explicit moment and conditioning
requirements, not universal prediction dominance.

\subsection{Relation to neighbouring frameworks}
\label{sec:disc:neighbours}

The method complements \BFP{}, robust M-estimation and GMM. \BFP{} is
preferable for posterior uncertainty over model space; PMM-FP is
preferable for fast frequentist estimation with direct efficiency
diagnostics under non-Gaussian residuals. Robust M-estimators
\citep{huber1981robust} target outliers and contamination, whereas
PMM-FP targets systematic skewness and kurtosis. The PATP preprint
\citep{zabolotnii2026patp} provides the broader signed-parity and
continuous-$\alpha$ foundation; PMM-FP is the statistical-modelling
specialisation of that programme with fixed FP power sets and explicit
estimating equations. The Cram\'er--Rao positioning and the GMM, Godambe and
quasi-likelihood connections are developed in the subsections below.

\subsection{Positioning relative to the Cram\'er--Rao bound}
\label{sec:disc:crb}

\paragraph{Three efficiency tiers.}
To place PMM-FP in the semiparametric landscape, consider the model
class
\begin{equation}
  \mathcal{M} = \bigl\{ f_\xi : \mathbb{E}[\xi]=0,\;
  \mathbb{E}[\xi^2]=\sigma^2<\infty,\;
  \mathbb{E}[\xi^4]<\infty \bigr\},
  \label{eq:semiparam-class}
\end{equation}
where $f_\xi$ is an infinite-dimensional nuisance parameter
\citep{bickel1993efficient,vanderVaart1998asymptotic}.
For any regular estimator in $\mathcal{M}$ the semiparametric
Cram\'er--Rao bound (CRB) is given by
$I^{*}(\theta)^{-1}$ \citep[Theorem~25.20]{vanderVaart1998asymptotic};
attaining it requires non-parametric reconstruction of $f_\xi$
\citep{newey1990efficient}, which PMM-FP deliberately avoids.

There are two efficiency tiers that are validated by the present
implementation and experiments.

\begin{enumerate}
  \item \emph{OLS-FP} ($g_2 = 1$): the Gaussian baseline, attained by
        all regular distributions when $\gamma_3 = 0$.
  \item \emph{Implemented PMM-FP} ($g_2^{\mathrm{class}} =
        1 - \gamma_3^2/(2+\gamma_4)$; Theorem~\ref{thm:T3a}): the
        numerically stable, closed-form skewness--kurtosis correction
        delivered by the estimator. Numerical values for the benchmark
        distributions are: Normal 1.000, Beta$(2,5)$ 0.811,
        Gamma$(\alpha=3)$ 0.667, Lognormal$(0.5)$ 0.612, Exp$(1)$
        0.500.
\end{enumerate}

\noindent
The semiparametric CRB lower bound $g_{\mathrm{LB}}$ is approximately
$0.333$ for both Gamma$(\alpha=3)$ and Lognormal$(0.5)$. The implemented
two-moment PMM-FP correction remains above that bound: $0.667$ and
$0.612$, respectively. This is expected for an estimator that avoids
non-parametric reconstruction of the error density. (These analytic tiers
use Lognormal with $\sigma=0.5$; the Monte Carlo study in
\S\ref{sec:simulation} uses the heavier-tailed $\sigma=1$ log-normal, for
which the robust variance reduction is correspondingly larger.) Exp$(1)$ and
Beta$(2,5)$ have support boundaries that depend on the location
parameter, so the standard parametric CRB comparison is not used for
those cases.

\paragraph{Honest summary of the efficiency position.}
The \emph{implemented, validated} PMM-FP is a numerically stable,
closed-form skewness--kurtosis correction that sits above
the semiparametric CRB by a definite and quantifiable margin; it is
not a near-efficient semiparametric estimator. Additional fractional
residual score functions remain a natural direction for future work,
but in the present data ranges they are treated as conditioning
diagnostics rather than as a source of claimed finite-sample efficiency
gains. Fully efficient semiparametric estimators
\citep{newey1990efficient,bickel1993efficient} require
non-parametric density estimation, which PMM-FP deliberately avoids in
favour of a closed-form moment score.

\subsection{Formal bridges: GMM, Godambe, quasi-likelihood}
\label{sec:disc:bridges}

The PMM-FP estimating equation fits naturally into three established
Western estimation frameworks; making this precise clarifies both the
method's scope and its limitations.

\paragraph{GMM form \citep{hansen1982gmm}.}
The score equation of PMM-FP (Theorem~\ref{thm:T3a}) can be written as
\[
  g_n(\theta) = \frac{1}{n} \sum_{v=1}^{n}
  \psi^{(\mathrm{FP})}\!\bigl(\xi_v;\theta\bigr)\,
  \nabla_\theta R(x_v;\theta) = 0,
\]
where $\psi^{(\mathrm{FP})}$ is a $K$-dimensional moment score built
from the fractional-polynomial residual basis. This is exactly the GMM
setup of \citet{hansen1982gmm} with $K$ moment conditions (one per
basis function). The optimal GMM weighting
$W^{*} = (\mathbb{E}[g_n g_n^\top])^{-1}$ recovers the coefficients
$c_k(\gamma_3, \gamma_4)$ that determine $g_2 = 1 -
\gamma_3^2/(2+\gamma_4)$. When the basis is just-identified (number of
moment conditions equals number of parameters) the solution is unique
and coincides with the classical method of moments; when it is
over-identified, Hansen's $J$-test provides a specification diagnostic.

\paragraph{Godambe information \citep{godambe1960}.}
For an arbitrary score function $\Psi(\xi;\theta)$, the Godambe
information is the sandwich matrix
\[
  I_G(\theta) = \mathbb{E}\!\bigl[\partial_\theta \Psi\bigr]^\top
  \Var[\Psi]^{-1}
  \mathbb{E}\!\bigl[\partial_\theta \Psi\bigr],
\]
which reduces to the efficient Fisher information $I^{*}(\theta)$
when $\Psi = \ell^{*}_\theta$. Kunchenko's construction
\citep[ch.~4]{kunchenko2002} derives $\psi^{(\mathrm{FP})}$ as the
$\Var[\Psi]$-minimiser within a class of moment score functions of
fixed structure $\partial_\theta R$. Therefore PMM-FP is
\emph{Godambe-optimal within its score class} but \emph{not
Fisher-optimal} in $\mathcal{M}$: $I_G(\theta) \preceq I^{*}(\theta)$,
with equality only when $\psi^{(\mathrm{FP})} = \ell^{*}_\theta$,
which does not hold for the four-moment fixed-structure score.

\paragraph{Quasi-likelihood \citep{wedderburn1974}.}
Wedderburn's quasi-likelihood equation uses only the first two moments
($\mathbb{E}[\xi]=0$, $\Var[\xi]=V(\mu)$) and yields a Godambe-optimal
estimator within the class of score functions \emph{linear} in $y_v -
\mu_v$. PMM-FP is not a quasi-likelihood method: it uses score
functions of higher degree in $\xi$, informed by the cumulants
$\gamma_3, \gamma_4$ beyond the mean--variance link. The two apparatuses
share a moment-criterion philosophy, but PMM-FP is strictly richer: it
exploits third- and fourth-cumulant information unavailable to
quasi-likelihood and thereby attains $g_2 < 1$ whenever $\gamma_3 \neq
0$.

\subsection{Limitations}

Several limitations are left visible. The fish-nutrition data in
\citet{hubin2026bfp} are a notational and methodological anchor, not a
replicated application. The submitted implementation uses the
second-order score; higher-order PMM-FP is theoretically supported but
requires stronger moments and stabler matrix estimation. Lean checks the
definitions, algebraic identities and conditional theorem structure,
while the full M-estimator CLT and Schur-complement theory remain
standard statistical inputs. Finally, small samples and very heavy tails
can make residual cumulants noisy, so bootstrap uncertainty, trimming
sensitivity and condition numbers should be reported before interpreting
$\hat g_2$ as an efficiency forecast.

\section{Conclusions}\label{sec:conclusions}

This paper provides an efficient, drop-in frequentist estimator for the
coefficients of a selected fractional-polynomial model. It has three
components: a default positive-power fractional residual basis (with a BD0
admissibility boundary for negative powers), a PMM estimating score with the
closed-form second-order variance ratio in~\eqref{eq:g2-pmm2}, and a
reproducible stored-results implementation that separates model selection
from score-based coefficient estimation and records the diagnostics needed
to check the reported tables and figures.

Three findings frame the contribution. First, the variance-reduction
result is validated on both simulated and real data: at a correctly
specified fixed model the PMM-FP coefficient variance equals
$g_2 = 1 - \gamma_3^2/(2+\gamma_4)$ times the OLS-FP variance, confirmed
by matched-basis Monte Carlo and, on the GBSG breast-cancer data, by a
bootstrap variance ratio of $0.53$ against the predicted $0.56$. Second,
the realised efficiency is stated honestly: the implemented two-moment
estimator sits a definite margin above the semiparametric
Cram\'er--Rao bound (for Gamma$(3)$, $g_2 = 0.67$ versus a bound of
$0.33$); the full negative-power basis is inadmissible for standard error
laws (condition~BD0 fails) and is reported only as a conditioning
diagnostic, not as a competing estimator. Third, the dominant
practical risk is not estimator efficiency but model-selection
uncertainty: when the BIC landscape over fractional powers is flat,
single-best confidence intervals under-cover, and frequentist model
averaging restores near-nominal coverage (from $0.75$ to $0.91$ across
our experiments).

The practical message is simple. Use PMM-FP when the FP mean structure is
scientifically plausible and residuals are visibly non-Gaussian. The
positive track is the stable default; the full track should be reported
only with BD0, inverse-moment and condition-number diagnostics. When
those diagnostics fail, the positive track still gives a transparent
efficiency correction without treating negative powers as harmless.

All scripts needed to rebuild the reported tables, figures, R session
information and Lean build evidence are listed in the supplement
\texttt{README.md}. The GBSG dataset is public through
\texttt{mfp}; the fish-nutrition data of \citet{hubin2026bfp} are not
included and are not used for a replication claim.

\section*{Declaration of generative AI and AI-assisted technologies in the manuscript preparation process}

During the preparation of this manuscript, the author used Anthropic
Claude Code for implementation and debugging of the R Monte Carlo
pipeline, data parsing and quality control, and language editing of the
English text. All AI-assisted content was reviewed and verified by the
author, who takes full responsibility for the results.

\section*{Funding}

This research did not receive any specific grant from funding agencies
in the public, commercial, or not-for-profit sectors.

\section*{Declaration of competing interest}

The author declares that he has no known competing financial interests
or personal relationships that could have appeared to influence the work
reported in this paper.

\section*{Data availability}

The German Breast Cancer Study Group data are publicly available through
the R package \texttt{mfp}, and the primary biliary cirrhosis data through
the R package \texttt{survival}. The reproducibility scripts, stored
simulation summaries, generated tables, figures, supplementary proof sketches
and Lean build evidence are provided in the public code supplement at
\href{https://github.com/SZabolotnii/Ku-PMM-FP-code-supplement}{github.com/SZabolotnii/Ku-PMM-FP-code-supplement}.

\bibliographystyle{plainnat}
\bibliography{pmm_fp}

\end{document}